\documentclass[twoside]{article}

\usepackage[accepted]{aistats2024}
\usepackage[round]{natbib}

\usepackage[ruled]{algorithm2e}
\usepackage{amsmath}
\usepackage{amssymb}
\usepackage{amsthm}
\usepackage{graphicx}
\usepackage{tikz}

\usepgflibrary{arrows}
\usetikzlibrary{shapes}
\usetikzlibrary{arrows.meta}
\usetikzlibrary{positioning}

\newtheorem{defn}{Definition}[section]
\newtheorem{exmp}[defn]{Example}
\newtheorem{lem}[defn]{Lemma}
\newtheorem{prop}[defn]{Proposition}
\newtheorem{thm}[defn]{Theorem}
\newtheorem{cor}[defn]{Corollary}

\newcommand{\an}{\text{an}}

\newcommand{\pa}{\text{pa}}

\newcommand{\indep}{\perp \!\!\! \perp}

\begin{document}

%

%

\twocolumn[

\aistatstitle{Faithful graphical representations of local independence}

\aistatsauthor{Søren Wengel Mogensen}

\aistatsaddress{Department of Automatic Control, Lund University} ]

\begin{abstract}
  Graphical models use graphs to represent conditional independence 
  structure in the
  distribution of a random vector. In stochastic processes, graphs may 
  represent so-called local 
  independence or conditional Granger causality. Under some regularity 
  conditions, a local independence graph implies a set of independences using a 
  graphical criterion 
  known as $\delta$-separation, or using its generalization, $\mu$-separation. 
  This 
  is a stochastic process analogue of 
  $d$-separation in DAGs. However, there may be 
  more independences than implied by this graph and this is a violation of 
  so-called \emph{faithfulness}. We characterize faithfulness in local 
  independence graphs and give a method to construct a faithful graph from any 
  local independence model such that the output equals the true graph when 
  Markov and faithfulness assumptions hold. We discuss various assumptions that 
  are weaker than faithfulness, and we explore different structure learning 
  algorithms and their properties under varying assumptions.
\end{abstract}

\section{Introduction}
\label{sec:intro}

Graphical models are widely used and so-called \emph{Markov 
properties} are essential as they describe how graphs encode conditional 
independence
\citep{lauritzen1996}.
While such Markov properties hold under fairly general conditions, it is 
well-understood that conditional independence models are too complicated to be 
described completely by these properties. One particular issue 
is the potential lack of \emph{faithfulness} such that the graph encodes a 
dependence which is not in the probability distribution 
\citep{spirtes2018search}.

In models of multivariate stochastic processes, tests of \emph{local 
independence} or 
\emph{Granger causality} may be used to learn a causal graph in which each node 
represents a coordinate process. Most prior work 
assumes that the causal graph is Markov and faithful with 
respect to the observed independences. This might not hold, even in the 
theoretical distribution from which we sample data. Moreover, when presented 
with real data, we need statistical tests of local independence and therefore 
wrong test results will also distort the output. 

In this paper, we characterize faithfulness and discuss 
a hierarchy of faithfulness assumptions that are relevant in this context. We 
describe differences between structure learning in DAG-based models and in 
stochastic process models. We compare different algorithms for use in 
stochastic process models and highlight how to minimize the impact of 
faithfulness issues. We start by defining the two independence relations that 
we will use.

\subsection{Local Independence}

Local independence is a ternary independence relation 
\citep{schweder1970,aalen1987,didelez2008} and we will use graphs to represent 
local independence in a multivariate stochastic process, analogously to how 
graphs may encode conditional independence in the distribution of a random 
vector.

The definition of local independence will depend on the class of stochastic 
processes we consider. We follow the definition in \cite{mogensenUAI2018}. Let 
$X_t = (X_t^1,\ldots,X_t^n)$ be a continuous-time stochastic process. We say 
that $X_t^i$ is a 
\emph{coordinate process}. We let $V=\{1,2,\ldots,n \}$. For $D\subseteq V$, we 
define $\mathcal{F}_t^D$ as the completed and right-continuous version of 
$\sigma(\{X_s^\alpha: s<t, \alpha\in D \})$.

\begin{defn}[Local independence]
	Let $\lambda_t = (\lambda_t^1,\ldots,\lambda_t^n)$ be a stochastic process. 
	Let $A,B,C\subseteq V$. We say that $X^B$ is \emph{locally independent} of 
	$X^A$ given $X^C$, or simply that $B$ is \emph{locally independent} of $A$ 
	given $C$, if for all $\beta\in B$ 
	
	\begin{align*}
		t \mapsto E(\lambda_t^\beta \mid \mathcal{F}_t^{A,C})
	\end{align*}
	
	has an $\mathcal{F}_t^{C}$-adapted version.
	\label{def:li}
\end{defn}

The above definition does not answer the important question: What should the 
$\lambda$-process be? This will depend on the class of processes. For 
stochastic differential equations, $\lambda$ is the drift 
\citep{mogensenUAI2018}. For point processes, it is the conditional intensity. 
We give a detailed point process example in Appendix \ref{app:hawkes}. The 
$\lambda$-process should essentially 
describe how the immediate evolution of the multivariate process depends on the 
past. If so, $B$ is locally independent of $A$ given $C$ if predicting the 
immediate future of $B$ can be done equally well using the past of process $C$ 
only or 
the past of processes $A$ and $C$.

\subsection{Granger Causality}

Granger causality \citep{granger1969investigating} is at times treated with 
some suspicion as it is said to not 
be `true' causality. In this paper, we only use Granger causality as an 
independence relation, analogously to how conditional independence is used in 
causal models of random vectors. In this way, tests of Granger causality can 
help us identify certain features of the underlying causal graph, and 
\emph{(conditional) Granger independence} would in fact be a better term
for our usage of 
Granger causality. In this context, $X = (X_t^1,\ldots,X_t^n)$ is a 
multivariate time series, that is, a stochastic process in discrete time. We 
let $X_{<t}^D$ denote the set $\{X_s^\alpha: s<t, \alpha\in D \}$.

\begin{defn}[Granger causality]
	Let $A,B,C\subseteq V$. We say that $X^B$ is \emph{Granger-noncausal} for 
	$X^A$ given $X^C$ if for all $t$ and all $\beta\in B$,
	
	\begin{align*}
		X_t^\beta \indep X_{<t}^A \mid X_{<t}^C
	\end{align*}
	
	where $\cdot \indep \cdot \mid \cdot$ denotes conditional independence.
	\label{def:granger}
\end{defn}

\begin{exmp}[VAR]
	As an example of a time series model, we consider a 
	\emph{vector-autoregressive process} of order 1. For each $t$, we have
	
	\begin{align*}
		X_t = A X_{t-1} + \varepsilon_t
	\end{align*}
	
	such that $(\varepsilon_t)$ is a sequence of independent random vectors. 
	Moreover, the entries of $\varepsilon_t$ are independent. In this case, the 
	zeroes of the $n\times n$ matrix $A$ encode which variables at time $t-1$ 
	directly influence the variables at time $t$. We can construct an intuitive 
	graphical representation with nodes $V = \{1,2,\ldots,n \}$ by including 
	the edge $\alpha\rightarrow\beta$ if and only if $A_{\beta\alpha}\neq 0$. 
	Assume now that $n = 4$, $V = \{1,2,3,4\}$, and 
	
	\begin{align*}
		A = \begin{bmatrix}
		a_{11}       & 0 & 0 & 0 \\
		a_{21}       & a_{22} & a_{23} & a_{2n} \\
		0       & 0 & a_{23} & 0 \\
		0       & a_{42} & a_{43} & a_{44}
		\end{bmatrix}
	\end{align*}
	
	where the entries of $A$ are nonzero if not indicated as zero in the above 
	equation. The corresponding graph is in Figure \ref{fig:exmp}.
	
	\begin{figure}
		\centering
												\begin{tikzpicture}[->,> = 
												latex',thick,scale=0.7]
												\tikzset{vertex/.style
													= 
													{shape=circle,draw,minimum 
														size=1.5em, 
														inner sep = 0pt}}
												\tikzset{edge/.style= {->,> = 
														latex',thick}}
												\tikzset{edgebi/.style= {<->,> 
												= 
														latex', 
														thick}}
												\tikzset{every 
													loop/.style={->,> = 
														latex',thick,min 
														distance=8mm, 
														looseness=5}}
												\tikzset{vertexFac/.style= 	
													{shape=rectangle,draw,minimum
													 
														size=1.5em, 
														inner sep = 
														0pt}}
												
												%
												\def\y{-2}						
												\node[vertex] (i) at (-4,0+\y) 	
												{$1$};
												\node[vertex] (a) at  	
												(-2,0+\y) {$2$};
												\node[vertex] (b) 	at  
												(2,0+\y) 	
												{$4$};
												\node[vertex] 	(u) 	
												at  	
												(0,1.25+\y)	{$3$};

												
												\draw[edge] (i) to 
												(a);
												\draw[edge] (a) to 
												(b);				
												\draw[edge] (u) to 
												(a);
												\draw[edge] (u) to 
												(b);
												\draw[edge, bend left] (b) to 
												(a);
												\draw[loop above] (a) to 
												(a);	
												\draw[loop above] (i) to 
												(i);
												\draw[loop right] (u) to 
												(u);
												\draw[loop above] (b) to 
												(b);

												\end{tikzpicture}
\caption{Graph from Example \ref{exmp:var}. In this graph, 4 is $\mu$-separated 
from 1 given $\{2,3,4\}$. Under the global Markov property, this implies that 
1 is Granger noncausal for 4 given $\{2,3,4\}$. This means that we are able 
to predict the present of variable 4, $X_t^4$, equally well using the past of 
processes 
$\{2,3,4\}$, $X_{<t}^{\{2,3,4\}}$, and using the past of processes 
$\{1,2,3,4\}$, 
$X_{<t}^{\{1,2,3,4\}}$. That is, conditionally on the past 
of $\{2,3,4\}$, the past of process 1 does not add any information on the 
present value of 4. On the other hand, 4 is not $\mu$-separated from 1 given 
$\{2,4\}$ as the path $1 \rightarrow 2 \leftarrow 3 \rightarrow 4$ is 
$\mu$-connecting.}
\label{fig:exmp}
	\end{figure}
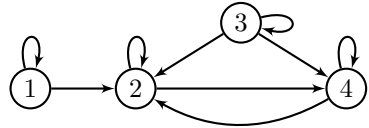
	
	\label{exmp:var}
\end{exmp}

\subsection{Graph Prerequisites}

A graph is an ordered pair $(V,E)$ where $V$ is a finite set of \emph{nodes} 
(also known as \emph{vertices}) and $E$ is a finite set of \emph{edges}. In 
this paper, we will mostly consider \emph{directed graphs} in which $E$ can be 
thought of as a subset of $V\times V$. For $\alpha,\beta\in V$, the edge
$\alpha\rightarrow\beta$ is in the graph if $(\alpha,\beta)\in E$. We 
always include all \emph{self-edges}, i.e., edges 
$\alpha\rightarrow\alpha$ for $\alpha\in V$.

For graphs $\mathcal{D}_1 = (V,E_1)$ and $\mathcal{D}_2 = (V,E_2)$, we say that 
$\mathcal{D}_1$ is a \emph{(proper) subgraph} of $\mathcal{D}_2$ if 
$E_1\subseteq E_2$ ($E_1\subsetneq E_2$), and we denote this by 
$\mathcal{D}_1\subseteq\mathcal{D}_2$ ($\mathcal{D}_1\subsetneq\mathcal{D}_2$). 
We also say that $\mathcal{D}_2$ is a \emph{(proper) supergraph} of 
$\mathcal{D}_1$. A \emph{walk}, $\omega$, is an ordered, alternating sequence 
of nodes and 
edges, $\alpha_1,e_1, \alpha_2, \ldots , \alpha_n, e_l, \alpha_{l+1}$, such 
that each edge is between its adjacent nodes. The \emph{length} of the walk 
$\omega$ is $l$. A \emph{path} is a walk such that no node is repeated. For 
nodes $\alpha,\beta\in V$, we say that a walk from $\alpha$ to $\beta$ is 
\emph{directed} if every edge points towards $\beta$, $\alpha\rightarrow \ldots 
\rightarrow\beta$. If there exists a directed walk from $\alpha$ to $\beta$, we 
say that $\alpha$ is an \emph{ancestor} of $\beta$. We let 
$\an_\mathcal{D}(\beta)$ denote the set of ancestors of $\beta$, and we let 
$\an_\mathcal{D}(B) = \cup_{\beta\in B} \an_\mathcal{D}(\beta)$. By convention, 
we say that a \emph{trivial walk} (a walk with no edges) is directed and 
therefore $B\subseteq \an_\mathcal{D}(B)$. The \emph{complete graph} on nodes 
$V$ is the graph $(V,E)$ such that $(\alpha,\beta)\in E$ for all $\alpha$ and 
$\beta$ such that $\alpha\neq \beta$.

We will use \emph{$\mu$-separation} to encode local 
independence or Granger noncausality. This is analogous to how $d$-separation 
in DAGs 
may encode conditional independence.

\begin{defn}[$\mu$-separation, \cite{mogensenUAI2018,Mogensen2020a}]
	Let $\mathcal{D}=(V,E)$ be a graph and let $A,B,C\subseteq V$. We say that 
	$B$ is 
	$\mu$-separated from $A$ given $C$ if there is no nontrivial walk in 
	$\mathcal{D}$ between any node 
	$\alpha\in A\setminus C$ and any node $\beta\in B$ such that all colliders 
	are in 
	$an_\mathcal{D}(C)$, no noncolliders are in $C$, and the final edge has a 
	head at $\beta$.
	\label{def:deltaSep}
\end{defn}

The notion of $\mu$-separation is a generalization of 
$\delta$-separation
\citep{didelez2000,didelez2008}.

\subsection{Independence Models}

In this paper, we will use an abstract \emph{independence model}, 
$\mathcal{I}$, which is 
simply a set of triples, $(A,B,C)$, $A,B,C\subseteq V$, and we say that this is 
an independence model \emph{over} $V$. Such an independence model 
may represent the local independences that hold in a multivariate, 
continuous-time stochastic process or the conditional Granger-noncausalities 
that hold in a discrete-time stochastic process, i.e., $(A,B,C) \in 
\mathcal{I}$ if and only if $B$ is locally independent 
of $A$ given $C$, for example. Using an abstract independence model, there is 
no need to distinguish between 
independence models representing local independences and independence models 
representing Granger noncausalities. In the remainder of the paper, we will 
often refer to both types of independences as simply `local independences'.

For a graph $\mathcal{D}$, we define $\mathcal{I}(\mathcal{D})$ as the set of 
triples 
$(A,B,C)$ such that $B$ is $\mu$-separated from $A$ given $C$ in $\mathcal{D}$. 
\emph{Markov and 
faithfulness properties} describe how 
$\mathcal{I}$ and $\mathcal{I}(\mathcal{D})$ are related.

\subsection{Markov Properties and Faithfulness}

\emph{Markov properties} describe how graphs encode independence by 
relating 
properties of a graph, $\mathcal{D}$, to an independence model, 
$\mathcal{I}$. We use the notation $\alpha \rightarrow_\mathcal{D}\beta$ to 
indicate that the edge $\alpha\rightarrow\beta$ is in $\mathcal{D}$.

\begin{defn}[Pairwise Markov property]
	We say that $\mathcal{I}$ satisfies the \emph{pairwise Markov property} 
	with 
	respect to $\mathcal{D}$ if for all $\alpha,\beta\in V$
	
	\begin{align*}
		\alpha\not\rightarrow_\mathcal{D}\beta\Rightarrow(\alpha,\beta,V\setminus
		\{\alpha \})\in \mathcal{I}.
	\end{align*}
	\label{def:pairwiseMarkov}
\end{defn}

\begin{defn}[Global Markov property]
	We say that $\mathcal{I}$ satisfies the \emph{global Markov property} with 
	respect to $\mathcal{D}$, or simply that $\mathcal{I}$ is \emph{Markov} 
	with respect to $\mathcal{D}$, if for all $A,B,C\subseteq V$,
	
	\begin{align*}
		(A,B,C) \in \mathcal{I}(\mathcal{D}) \Rightarrow (A,B,C) \in 
		\mathcal{I}.
	\end{align*}
	
	The global Markov property may also be written as 
	$\mathcal{I}(\mathcal{D})\subseteq\mathcal{I}$.
	\label{def:globalMarkov}
\end{defn}

The global and pairwise Markov properties are equivalent under fairly general 
assumptions, see, e.g., \cite{didelez2000, 
didelez2008,eichler2012,mogensenUAI2018} for related results in different model 
classes. Some of these results restrict the sets $A$,$B$, and $C$, e.g., such 
that $B\subseteq C$ in our notation.

\begin{defn}[Faithfulness]
	We say that $\mathcal{I}$ is \emph{faithful} with respect to $\mathcal{D}$ 
	if for all $A,B,C\subseteq V$,
	
	\begin{align*}
	(A,B,C) \in 
	\mathcal{I} \Rightarrow (A,B,C) \in \mathcal{I}(\mathcal{D}),
	\end{align*}
	
	that is, if $\mathcal{I}\subseteq \mathcal{I}(\mathcal{D})$.
	\label{def:faith}
\end{defn}

Note that, in our terminology, faithfulness corresponds to the statement 
$\mathcal{I} \subseteq \mathcal{I}(\mathcal{G})$, not to the stronger statement 
$\mathcal{I} = \mathcal{I}(\mathcal{G})$.

\subsection{Structure Learning}
\label{ssec:structLearning}

There is a large literature on structure learning from multivariate stochastic 
processes, often assuming \emph{causal sufficiency}, i.e., that every 
relevant coordinate process is observed, and assuming some specific parametric 
or semiparametric class of stochastic processes. We will also make 
the assumption of causal sufficiency in 
this paper, however, we will take a nonparametric approach. For 
parametric model classes, and assuming causal 
sufficiency, one may also, e.g., use methods that are specific to the model 
class to learn a causal graph from data. Our approach is completely 
nonparametric in that it only uses tests of local independence. Examples 
\ref{exmp:hawkes} and \ref{exmp:var} are therefore mostly 
meant as an illustration.

In the 
next section, we give a characterization of faithfulness which allows us to 
construct faithful representations of local independence models.

\subsubsection{Causal Interpretation}
\label{sssec:cauInter}

Structure learning is often done from a causal perspective. The causal 
interpretation will also depend on the model class. In this paper, we assume 
that 
$\mathcal{D}$ is a \emph{causal graph} which summarizes the cause-effect 
relations between the coordinate processes of the system. The exact meaning of 
this is discussed by, e.g., \cite{eichler2007,roysland2023graphical}.

\section{Transitivity Conditions}

We define a set of \emph{transitivity conditions}.

\begin{defn}[Transitivity conditions]
	Let $\mathcal{D} = (V,E)$ and let $\mathcal{I}$ be an independence model 
	over $V$.
	Let $C\subseteq V$. We say that $\mathcal{I}$ is $C$-transitive with 
	respect to $\mathcal{D}$ if for each edge
	$\alpha\rightarrow\beta$ in $\mathcal{D}$, conditions D0-D3 hold.
	\begin{enumerate}
		\item[D0] if $\alpha\notin C$, then $(\alpha,\beta,C)\notin 
		\mathcal{I}$,
		\item[D1] if $\alpha\notin C$, then for all $\gamma$: \\ 
		$(\gamma,\alpha,C)\notin \mathcal{I}(\mathcal{D}) 
		\Rightarrow (\gamma,\beta,C)\notin \mathcal{I}$,
		\item[D2] if $\alpha\notin C,\beta\in C$, then for all 
		$\gamma,\delta$:\\
		$(\gamma,\beta,C) \notin \mathcal{I}(\mathcal{D}), 
		(\alpha,\delta,C)\notin \mathcal{I}(\mathcal{D}) \Rightarrow 
		(\gamma,\delta,C)\notin\mathcal{I}$,
		\item[D3] if $\alpha\notin C$, then for all $\gamma$:\\
		$(\alpha,\gamma,C)\notin \mathcal{I}(\mathcal{D})\Rightarrow 
		(\beta,\gamma,C)\notin \mathcal{I}$.
	\end{enumerate}
	
	We say that an independence model $\mathcal{I}$ is \emph{transitively 
	closed} with respect to a graph 
	$\mathcal{D}$ if $\mathcal{I}$ is $C$-transitive with respect to 
	$\mathcal{D}$ for all $C\subseteq V$.
	\label{def:transitivity}
\end{defn}

A simpler version of the conditions in Definition \ref{def:transitivity} 
are also found in 
\cite{Mogensen2020a} where the authors used them to prove that every Markov 
equivalence class of partially observed local independence graphs have a 
greatest element. We can recover their version by using $\mathcal{I} = 
\mathcal{I}(\mathcal{D})$ in the above definition. For our result, the 
generalization is important as it connects an arbitrary independence model, 
$\mathcal{I}$, to a graphical representation, $\mathcal{D}$.

The conditions in Definition \ref{def:transitivity} are in a certain sense 
rewriting the definition of $\mu$-separation. This has three purposes. First, 
this 
assigns faithfulness violations to specific edges that can be removed to obtain 
faithful representations (Subsection \ref{ssec:trim}). Second, it allows us to 
construct a 
(nontrivial) faithful representation directly from the independence model 
(Section \ref{sec:characFaith}). 
Third, we will reformulate these conditions slightly to see that they 
correspond to different 
notions of faithfulness (Section \ref{sec:weakFaith}).

\begin{prop}
	The independence model $\mathcal{I}(\mathcal{D})$ is transitively closed 
	with respect to the
	 graph $\mathcal{D}$.
	\label{prop:graphTransitive}
\end{prop}

\begin{prop}
	Let $\mathcal{I}_1\subseteq \mathcal{I}_2$. If $\mathcal{I}_2$ is 
	transitively closed with respect to $\mathcal{D}$, then $\mathcal{I}_1$ is 
	transitively closed with respect to $\mathcal{D}$.
	\label{prop:transitiveInclusion}
\end{prop}

\section{Characterization of Faithfulness}
\label{sec:characFaith}

Definition \ref{def:transitivity} gives a 
characterization of faithfulness as described in the next theorem.

\begin{thm}
	An independence model $\mathcal{I}$ is transitively closed with respect to 
	a 
	graph $\mathcal{D}$ if and only if $\mathcal{I}$ is faithful with 
	respect to $\mathcal{D}$.
	\label{thm:characFaith}
\end{thm}

The above characterizes the set of graphs, $\mathcal{D}$, 
that are faithful with respect to the independence model $\mathcal{I}$. 
However, the conditions in Definition \ref{def:transitivity} use both 
$\mathcal{I}$ and $\mathcal{I}(\mathcal{D})$, and it is therefore not 
immediately clear 
how to construct these graphs if we only have access to the independence 
model $\mathcal{I}$. The next definition defines a graph from an independence 
model, $\mathcal{I}$, only, and Theorem \ref{thm:edgeFaith} proves that  
$\mathcal{I}$ in fact is faithful with respect to the graph that we obtain from 
the 
definition.

\begin{defn}[Edge-transitive graph]
	Let $\mathcal{I}$ be an independence model over $V$. We define a graph 
	$\mathcal{F}_I = (V,E_I)$ by including the edge $\alpha\rightarrow\beta$, 
	$\alpha,\beta\in V$, $\alpha\neq\beta$, if and only if E0-E3 hold for 
	all $C$.
	\begin{enumerate}
		\item[E0] if $\alpha\notin C$, then $(\alpha,\beta,C)\notin 
		\mathcal{I}$,
		\item[E1] if $\alpha\notin C$, then for all $\gamma$: \\ 
		$(\gamma,\alpha,C)\notin \mathcal{I} 
		\Rightarrow (\gamma,\beta,C)\notin \mathcal{I}$,
		\item[E2] if $\alpha\notin C,\beta\in C$, then for all 
		$\gamma,\delta$:\\
		$(\gamma,\beta,C) \notin \mathcal{I}, 
		(\alpha,\delta,C)\notin \mathcal{I} \Rightarrow 
		(\gamma,\delta,C)\notin\mathcal{I}$,
		\item[E3] if $\alpha\notin C$, then for all $\gamma$:\\
		$(\alpha,\gamma,C)\notin \mathcal{I}\Rightarrow 
		(\beta,\gamma,C)\notin \mathcal{I}$.
	\end{enumerate}
	
	For an independence model, $\mathcal{I}$, we say that $\mathcal{F}_I$ 
	defined 
	above is the \emph{edge-transitive graph} corresponding to $\mathcal{I}$.
	\label{def:edgetransitive}
\end{defn}

\begin{thm}
	The independence model $\mathcal{I}$ is faithful with respect to 
	$\mathcal{F}_I$.
	\label{thm:edgeFaith} 
\end{thm}

We say that an independence model, $\mathcal{I}$, is \emph{graphical} if there 
exists a graph $\mathcal{D}$ such that $\mathcal{I} = 
\mathcal{I}(\mathcal{D})$. The following proposition simply states that if the 
independence model is 
Markov and faithful with respect to a graph, i.e., is graphical, then 
$\mathcal{F}_I$ as defined in Definition \ref{def:edgetransitive} is equal to 
$\mathcal{D}$.

\begin{prop}
	Assume $\mathcal{I}$ is graphical, that is, $\mathcal{I} = 
	\mathcal{I}(\mathcal{D})$ for a graph $\mathcal{D}$. In this case, 
	$\mathcal{F}_I = \mathcal{D}$.
	\label{prop:graphicalI}
\end{prop}

Any independence model is faithful with respect to the empty graph, and it is 
useful to introduce the concept of \emph{maximal faithfulness}. We say 
that $\mathcal{I}$ is \emph{maximally faithful} with respect 
to $\mathcal{D}$ if it is faithful with respect to $\mathcal{D}$ and it is not 
faithful with respect to any proper 
supergraph of $\mathcal{D}$.

\section{Weaker Notions of Faithfulness}
\label{sec:weakFaith}

The Markov condition holds under fairly general assumptions, however, some 
version of a faithfulness-like assumption is needed for structure learning. It 
is possible to define such notions that are weaker than faithfulness, yet 
useful in the context of structure learning (in DAG-based models, see, e.g., 
\cite{zhang2008detection, ramsey2006adjacency}). In local independence models, 
\cite{mogensen2020causal} gives the 
following definition.

\begin{defn}[Ancestor faithfulness, \cite{mogensen2020causal}]
	We say that $\mathcal{I}$ is \emph{ancestor faithful} with respect to 
	$\mathcal{D}$ if, for all $A,B$, and $C$ such that $A\not\subseteq C$, the 
	existence of a 
	directed and 
	$\mu$-connecting path from $A$ to $B$ given $C$ implies $(A,B,C)\notin 
	\mathcal{I}$. 
	\label{def:ancFaith}
\end{defn}

The following is a weaker notion than that of ancestor faithfulness.

\begin{defn}[Parent faithfulness]
	We say that $\mathcal{I}$ is \emph{parent faithful} with respect to 
	$\mathcal{D}$ if, for all $A,B$, and $C$ such that $A\not\subseteq 
	C$, the existence of a 
	directed edge $\alpha\rightarrow\beta$ such that $\alpha\in A$ and 
	$\beta\in B$ implies $(A,B,C)\notin \mathcal{I}$.
	\label{def:paFaith}
\end{defn}

We say that $\beta$ is \emph{inseparable} from $\alpha$ if there is no 
$C\subseteq V\setminus \{\alpha \}$ such that $(\alpha,\beta,C)\in 
\mathcal{I}(\mathcal{D})$.
Parent faithfulness can be seen as an analogue of adjacency faithfulness in 
DAG-based models: In a DAG, nodes are inseparable if and only if they are 
adjacent. In a local independence graph, a node $\beta$ is inseparable from a 
node $\alpha$ if and only if the edge $\alpha\rightarrow\beta$ is in the graph. 
One should note that the notion of inseparability is symmetric in DAGs, but 
asymmetric in local independence graphs. This means that in a local 
independence graph, $\alpha$ need not be inseparable from $\beta$ even if 
$\beta$ is inseparable from $\alpha$.

It may be that faithfulness is not violated, however, only closed to being 
violated. If so, learning methods that only assume weaker notions 
of 
faithfulness may show better performance \citep{ramsey2006adjacency, 
zhalama2017weakening}.

We define an even weaker faithfulness-like assumption.

\begin{defn}[Parent dependence]
	We say that $\mathcal{I}$ satisfies \emph{parent dependence} with respect 
	to $\mathcal{D}$, if $\alpha\rightarrow\beta$ implies $(\alpha,\beta, 
	\beta) \notin 
	\mathcal{I}$ for all $\alpha\neq \beta$.
\end{defn}

\subsection{Causal Minimality}

Faithfulness, and similar assumptions, are common for structure learning. In 
the context, of local independence there is a far weaker notion which is in 
fact sufficient for structure learning.

The concept of a maximally faithful graph is essentially dual to the concept of 
\emph{causal minimality}. We say that $\mathcal{I}$ is \emph{causally minimal} 
with respect to $\mathcal{D}$ if it is Markov with respect to $\mathcal{D}$ and 
there is no proper subgraph of $\mathcal{D}$, $\mathcal{D}'$, such that 
$\mathcal{I}$ is Markov with respect to $\mathcal{D}'$.
\citep{peters2017}. In symbols, $\mathcal{I}(\mathcal{D}) \subseteq 
\mathcal{I}$ 
and there is no $\mathcal{D}_0 \subsetneq \mathcal{D}$ such that 
$\mathcal{I}(\mathcal{D}_0) \subseteq \mathcal{I}$. Causal minimality is also 
known as \emph{minimal Markovness} \citep{sadeghi2017faithfulness}.

\begin{prop}
	Let $\mathcal{I}$ be an independence model and $\mathcal{D}$ be a graph. If 
	$\mathcal{I}$ is faithful with respect to $\mathcal{D}$, then it is 
	ancestor faithful with respect to $\mathcal{D}$. If $\mathcal{I}$ is 
	ancestor faithful with respect to $\mathcal{D}$, then it is parent 
	faithful with respect to $\mathcal{D}$. If $\mathcal{I}$ is parent faithful 
	and Markov with respect 
	to $\mathcal{D}$, then $\mathcal{I}$ is causally minimal with respect to 
	$\mathcal{D}$.
	\label{prop:hierarchyFaith}
\end{prop}

Definition \ref{def:edgetransitive} allows us to construct a faithful graph 
from an independence model (Theorem \ref{thm:edgeFaith}). We can also directly 
construct 
a causally minimal graph. For an independence model, $\mathcal{I}$, we define a 
graph, $\mathcal{D}_I$, such that $\alpha\rightarrow\beta$ is in 
$\mathcal{D}_I$ if and only if $(\alpha,\beta,V\setminus\{\alpha 
\})\notin\mathcal{I}$ and we say that $\mathcal{D}_I$ is the \emph{induced 
local 
independence graph} corresponding to $\mathcal{I}$.

\begin{prop}[\cite{mogensenThesis2020}]
	Assume equivalence of pairwise and global Markov properties. The induced 
	local independence graph corresponding to 
	$\mathcal{I}$, $\mathcal{D}_I$, is causally minimal with respect to 
	$\mathcal{I}$.
	\label{prop:inducedCauMin}
\end{prop}

\begin{prop}
	The graph $\mathcal{D}_I$ is the only causally minimal graph with respect 
	to $\mathcal{I}$.
	\label{prop:cauMinUnique}
\end{prop}

In other words, assuming the equivalence of pairwise and global Markov 
properties, $\mathcal{I}$ is causally minimal with respect to $\mathcal{D}$ if 
and only if $\alpha\rightarrow\beta$ is in $\mathcal{D}$ exactly when 
$(\alpha,\beta,V\setminus \{\alpha\} )\notin \mathcal{I}$.

Theorem \ref{thm:detectable} argues that violations of faithfulness, in 
principle, are detectable under Markov and causal minimality assumptions 
(Appendix \ref{app:detectable}).

The next proposition uses \emph{asymmetric graphoid properties} that hold in 
local independence models, see, e.g., \cite{didelezUAI2006, mogensenUAI2018} 
and Appendix \ref{sec:asym}.

\begin{prop}
	Assume $\mathcal{I}$ is causally minimal with respect to $\mathcal{D}$. If 
	$\alpha\notin pa_\mathcal{D}(\beta)$, $\alpha\neq \beta$, and 
	$\pa_\mathcal{D}(\beta)\subseteq C$, then $(\alpha,\beta, 
	C)\in\mathcal{I}$. 
	Assume that $\mathcal{I}$ satisfies left weak union, left decomposition, 
	and left 
	contraction, that we 
	have equivalence of pairwise and global Markov properties, 
	$\alpha\in pa_\mathcal{D}(\beta)$, $\alpha\neq \beta$, and that 
	$\pa_\mathcal{D}(\beta)\setminus \{ \alpha\} \subseteq C$. In this case,
	$(\alpha,\beta, 
	C)\notin\mathcal{I}$.
	\label{prop:superCpa}
\end{prop}

\subsection{Hierarchy of Faithfulness Assumptions}
\label{ssec:hierarchyFaith}

In this subsection, we rewrite the conditions in Definition 
\ref{def:transitivity} to illustrate how they correspond to different 
faithfulness assumptions. We first define the notion of \emph{trek 
faithfulness}.

\begin{defn}[Trek faithfulness]
	We say that a walk is a \emph{trek} if it has no colliders. We say that 
	$\mathcal{I}$ is \emph{trek faithful} with respect to $\mathcal{D}$, if for 
	all disjoint $A$, $B$, and $C$, the existence of a $\mu$-connecting trek 
	from $A$ to $B$ given $C$ implies $(A,B,C)\notin \mathcal{I}$.
	\label{def:trekfaith}
\end{defn}

It is immediate that faithfulness implies trek faithfulness, and that trek 
faithfulness implies ancestor faithfulness, noting that a directed walk is also 
a trek.

Lemma \ref{lem:conditions} reformulates the conditions from Definition
\ref{def:transitivity} to provide an equivalent set of conditions. These 
conditions correspond to the hierarchical nature of the faithfulness 
conditions: D0 is equivalent with parent faithfulness, the combination of D0 
and D1' is equivalent with ancestor faithfulness, and the combination of D0, 
D1', and D3' is equivalent with trek faithfulness. The combination of D0, D1', 
D2, and D3' is equivalent with faithfulness. This is the content of Theorem 
\ref{thm:conditionsFaith}.

\begin{lem}
	For an edge $\alpha\rightarrow\beta$ and a set $C$, we define the following 
	conditions.
	
	\begin{itemize}
		\item[D1'] If there is a directed path which is $\mu$-connecting 
		from $\gamma$ to $\alpha$ given $C$ in $\mathcal{D}$, then 
		$(\gamma,\beta,C)\notin \mathcal{I}$.
		\item[D3'] If there is a trek which is $\mu$-connecting from $\alpha$ 
		to $\gamma$ given $C$ in $\mathcal{D}$, then $(\beta,\gamma,C)\notin 
		\mathcal{I}$.
		\label{lem:conditions}
	\end{itemize} 
	
	An independence model $\mathcal{I}$ and a graph $\mathcal{D}$ satisfy D0, 
	D1, 
	D2, and D3 for every edge in $\mathcal{D}$ and set $C$ if and only if they 
	satisfy D0, 
	D1', 
	D2, and D3' for every edge in $\mathcal{D}$ and set $C$.
	\label{lem:reformulate}
\end{lem}

\begin{thm}
	Let $\mathcal{I}$ be an independence model which satisfies left and right 
	decomposition, and let $\mathcal{D}$ be a graph. 
	
	\begin{itemize}
		\item Condition D0 holds for all edges $\alpha\rightarrow\beta$ in 
		$\mathcal{D}$ and sets $C$ if and only if $\mathcal{I}$ is parent 
		faithful with 
		respect to $\mathcal{D}$.
		\item Conditions D0 and D1' hold for all edges $\alpha\rightarrow\beta$ 
		in $\mathcal{D}$ and sets $C$ if and only if $\mathcal{I}$ is ancestor 
		faithful with 
		respect to $\mathcal{D}$.
		\item Conditions D0, D1', and D3' hold for all edges 
		$\alpha\rightarrow\beta$ in $\mathcal{D}$ and sets $C$ if and only if 
		$\mathcal{I}$ 
		is trek faithful with respect to $\mathcal{D}$.
	\end{itemize}
	\label{thm:conditionsFaith}
\end{thm}

\section{Structure Learning}
\label{sec:struct}

In graphical structure learning, the task is to recover a graphical 
representation from tests of local independence. In this section, we describe 
how the above theory relates to structure learning algorithms. It 
is common to assume faithfulness in the context of structure learning, see, 
e.g., \cite{meek2014toward,mogensenUAI2018,absar2021discovering} for examples 
in 
structure learning 
based on local 
independence/Granger 
noncausality. \cite{mogensen2020causal} uses a weaker notion of faithfulness. 

We assume causal sufficiency except in Appendix \ref{app:partialObs} where we 
describe some results assuming only partial observation. 
As is common in the literature, we will at times assume that we have access to 
an \emph{independence oracle}, i.e., instead of inputting, e.g., $p$-values 
from 
tests of local independence, our algorithm simply has access to the actual 
independence model and therefore always gets the right answer to an 
independence query. This is mostly done to separate algorithmic issues from 
testing issues. In practical applications of the learning algorithms, the test 
$(\alpha,\beta,C)\in \mathcal{I}$ is replaced by a $p$-value and a significance 
threshold.

\subsection{Comparison with DAG-based Models}

There is a large literature on learning causal graphs based on 
tests of conditional independence (see \cite{spirtes2018search} and references 
therein). One example of an algorithm is the 
PC-algorithm \citep{spirtes1993}. In the adjacency phase of this algorithm, 
larger and larger conditioning sets are used to look for separating sets. One 
motivation is to use tests with small conditioning sets to achieve larger power 
of the statistical tests \citep{spirtes2018search}. 
\cite{meek2014toward} and 
\cite{absar2021discovering} proceed by 
checking larger and larger sets of potential separating sets and remove an edge 
when one is found, essentially using this basic idea of the PC-algorithm in 
the stochastic process-setting. However, there are a number of important 
differences between 
constraint-based learning in DAG-based models and constraint-based learning in 
stochastic process models. First, in the case, of DAG-based model several 
graphs may 
encode the same conditional independences. On the other hand, for stochastic 
processes and under quite general assumptions the causal graph is actually 
identified from the local independence model (see Section 
\ref{ssec:learnCausalMin}). 
Second, the set $V\setminus \{ \alpha\}$ $\mu$-separates $\beta$ from 
$\alpha$ if and only if $\alpha\rightarrow\beta$ is not in the graph. For this, 
we do not need to know the graph and essentially this means that we can 
construct a separating set, if one exists, without any knowledge of the 
graph.

\subsection{The CA-algorithm}
\label{ssec:ca}

We briefly describe the CA-algorithm from \cite{meek2014toward}. This is also 
similar 
to the algorithm in \cite{absar2021discovering}. In this algorithm, for each 
ordered pair $(\alpha,\beta)$, larger and larger conditioning sets are tried to 
find a 
separating set, i.e., a set, $C$, such that $(\alpha,\beta,C)\in \mathcal{I}$. 
This is similar to the classical PC-algorithm for DAGs \citep{spirtes1993}. The 
details of the algorithm can be found in \cite{meek2014toward}.

\subsection{The CS-algorithm}
\label{ssec:cs}

The CS-algorithm (\emph{causal screening}) was introduced in 
\cite{mogensen2020causal} as a fast 
screening approach for partially observed systems. In its first step, it tests 
$(\alpha,\beta,\beta) \in \mathcal{I}$ for all ordered pairs $(\alpha,\beta)$. 
In its second step, it tests $(\alpha,\beta, 
\pa_{\mathcal{D}_1}(\beta))\in\mathcal{I}$ where $\mathcal{D}_1$ is the output 
from the first step. The idea is to use a superset of the actual parent set of 
$\beta$ as 
a conditioning set.

\begin{prop}
	In the oracle case, the CS-algorithm (Algorithm \ref{algo:CS}) outputs the 
	true graph under Markov 
	and parent faithfulness assumptions.
\end{prop}

\begin{proof}
	If $\alpha\rightarrow\beta$ is not in the true graph, then it is also not 
	in the output (Proposition 2 in the supplementary material of 
	\citep{mogensen2020causal}). If is in the true graph, then 
	parent faithfulness implies that it is also in the output.
\end{proof}

\begin{prop}
	Assume causal sufficiency, left weak union, left decomposition, and left 
	contraction of 
	$\mathcal{I}$, and equivalence of pairwise and global Markov properties. If 
	$\mathcal{I}$ is causally minimal with respect to 
	$\mathcal{D}$ and satisfies parent dependence with respect to 
	$\mathcal{D}$, then causal screening outputs $\mathcal{D}$ in the oracle 
	setting.
	\label{prop:CSoutput}
\end{prop}

Many other algorithms will only be correct in the oracle case under stronger 
assumptions, one reason being that they test more `small' sets which may lead 
to 
a faulty edge removal due to a violation of faithfulness.

\begin{algorithm}
	\SetKwInOut{Input}{input}
	\SetKwInOut{Output}{output}
	\Input{$\mathcal{I}$ over $V$ such that $\vert V \vert = n$}
		\For{$\beta \in V$}{
			\For{$\alpha \in V\setminus \{\alpha\}$}{
					\If{$(\alpha,\beta,\beta)\in \mathcal{I}$}
					{
						$E \gets E \setminus \{\alpha\rightarrow \beta \}$\;
						$\mathcal{D} \gets (V,E)$\;
					}	}
				}	
	\For{$\beta \in V$}{
		\For{$\alpha \in \pa_\mathcal{D}(\beta)$}{
			\If{$(\alpha,\beta,\pa_\mathcal{D}(\beta)\setminus \{\alpha \})\in 
			\mathcal{I}$}
			{
				$E \gets E \setminus \{\alpha\rightarrow \beta \}$\;
				$\mathcal{D} \gets (V,E)$\;
			}	}
		}
			$\mathcal{D}_{cs} \gets \mathcal{D}$\;
			\Output{$\mathcal{D}_{cs}$}
			\caption{Causal screening algorithm (CS)}
			\label{algo:CS}
		\end{algorithm}

\subsection{Learning with Minimal Assumptions}
\label{ssec:learnCausalMin}

If we take the Markov property for granted, the four conditions outlined above, 
faithfulness, ancestor faithfulness, parent faithfulness, and causal 
minimality, are in this list ordered from strongest to weakest. In this 
subsection, we assume the weakest condition, that of causal minimality, to 
discuss how structure learning can be achieved with this minimal assumption.

Under 
Markov and causal minimality assumptions, the induced local independence graph, 
$\mathcal{D}_I$, equals the true graph (in 
the oracle case): Let $\mathcal{D}$ be the true graph such 
that $\mathcal{D}$ and $\mathcal{I}$ are causally minimal. The graph 
$\mathcal{D}_I$ satisfies the pairwise Markov property by definition. Under 
equivalence of pairwise and global Markov properties, we have that 
$\mathcal{D}_I$ is Markov with respect to $\mathcal{D}$. If $e$ is in 
$\mathcal{D}_I$, then $(\alpha,\beta,V\setminus \{\alpha \}) \notin 
\mathcal{I}$. Using Markovness, $e$ must be in $\mathcal{D}$ as well, 
so $\mathcal{D}_I\subseteq \mathcal{D}$. If $e$ is not in $\mathcal{D}_I$, then 
$e$ also not in $\mathcal{D}$ due to causal minimality.  The CM-algorithm 
(Causal Minimality) is the algorithm which outputs a graph, $\mathcal{D}_{cm}$, 
such that $\alpha\rightarrow_{\mathcal{D}_{cm}}\beta$ if and only if 
$(\alpha,\beta,V\setminus \{\alpha\})\notin \mathcal{I}$.

However, 
the above may not be practical if there are many 
coordinates processes as this may require very large conditioning sets, 
$V\setminus \{\alpha\}$, and 
therefore tests with poor performance. On the other hand, it avoids many tests, 
all of which have a risk of faithfulness violations, or near-violations, which 
leads to worse output. These observations may motivate the use of the 
CS-algorithm in Subsection \ref{ssec:cs}.

If we instead test all subsets, and include $\alpha\rightarrow\beta$ if and 
only if 
there is no separating set, we obtain a subgraph of the true graph, only 
assuming Markovness  (see details in Algorithm \ref{algo:dSGS} in Appendix 
\ref{app:dSGS}). This algorithm 
is a local independence version of the SGS 
algorithm \citep{spirtes1993}. If we assume that 
there are at most $k$ parents, we test 
all subsets of size at most $k$. This also returns a subgraph of the true graph 
under the Markov assumption. Appendix \ref{app:dSGS} defines this algorithm and 
gives states this result formally.

\subsection{Learning and Faithfulness}
\label{ssec:learnFaith}

The previous section describes general structure learning algorithms. Appendix 
\ref{app:learnFaith} connects structure learning with the faithfulness results 
in the previous 
	sections. In Subsection \ref{ssec:learnEdgeTransitive}, we consider the 
	\emph{edge-transitive graph}. In Subsection \ref{ssec:trim}, we argue 
	that one may trim the output of a learning algorithm to obtain a faithful 
	representation.

\begin{figure*}
	\includegraphics[scale = .7]{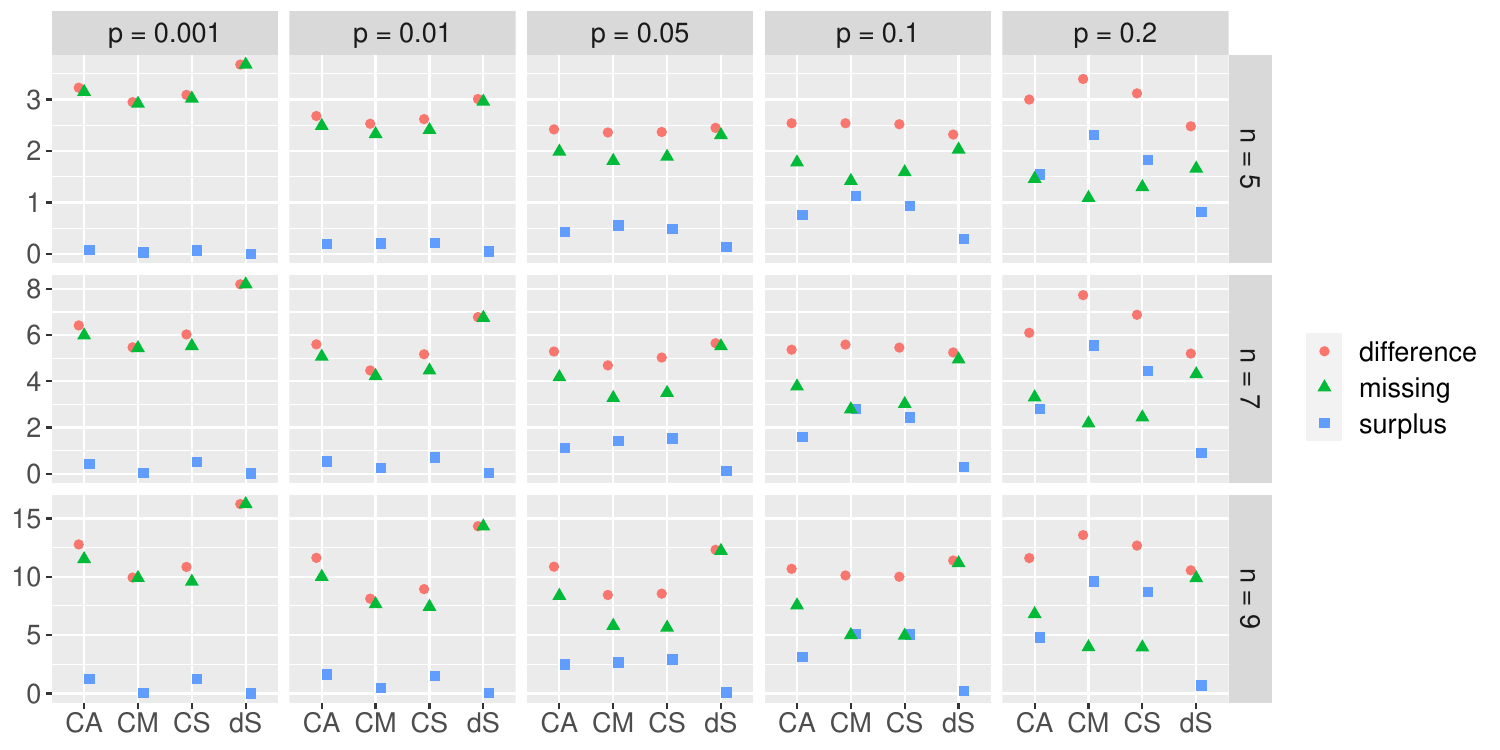}
	\caption{Comparison of algorithms. Points indicate mean over $M = 100$ 
		repetitions. Red circles indicate mean difference between true graph 
		and 
		output graph. Green triangles indicate mean number of surplus edges, 
		and 
		blue squares indicate mean number of missing edges. Section 
		\ref{sec:num} 
		explains this experiment in more detail.}
	\label{fig:num}
\end{figure*}

\section{Numerical Examples}
\label{sec:num}

We compare the algorithms to investigate their properties when using data. We 
repeatedly generated a true graph, $\mathcal{D} = (V,E)$, and observations from 
a corresponding 
VAR(1)-process. Using tests of Granger causality, we computed an output graph, 
$\mathcal{D}_a = (V,E_a)$ for each algorithm $a$. We 
computed the \emph{surplus edges}, $E_a\setminus E$, the \emph{missing edges} 
$E\setminus E_a$, and the \emph{difference}, $(E_a\setminus E) \cup (E\setminus 
E_a)$ between $\mathcal{D}$ and 
$\mathcal{D}_a$. In Figure 
\ref{fig:num}, we report the mean number of surplus edges, the mean number of 
missing edges, and the mean of $\vert 
(E_a\setminus E) \cup (E\setminus 
E_a)\vert $ for each algorithm $a$, and for different values of significance 
threshold and $n = \vert V\vert $. More details are in Section 
\ref{app:sim}.

The dSGS algorithm is seen to have the lowest number of surplus edges which is 
not surprising as it tests every possible set and removes the edge if any test 
is nonsignificant. We also know from Proposition \ref{prop:dSGSsubgraph} that, 
in the oracle case, it outputs a subgraph of the true graph under minimal 
assumptions.

We observe that the CM-algorithm, simply using a single test for each ordered 
pair $(\alpha,\beta)$ does surprisingly well, e.g., in comparison with the 
CA-algorithm. An important point is the fact that using more tests increase the 
risk of making errors due to faithfulness violations or near-violations. It is 
essential that, in this context, the set $V\setminus \{ \alpha\}$ always 
separated 
$\beta$ from $\alpha$ if such separation is possible, and this can be tested 
with no prior knowledge of the graph. For this reason, testing smaller 
conditioning sets is not needed, at least for $n$ of moderate size.

A key weakness of the CM-algorithm is, of course, the fact that it uses
large conditioning sets for large values of $n$ and such tests are expected to 
have low power. As a remedy, one may use the CS-algorithm which tries to reduce 
the 
set of potential parents. We see that the CS- and CM-algorithms have similar 
performances, and the CS-algorithm may be viable alternative for large values 
of $n$.

Appendix \ref{app:partialObs} provides additional results in the case of 
partial 
observation.

\section{Discussion}

Constraint-based learning in stochastic processes is still lacking some of the 
tools that are available for constraint-based learning from random vectors. 
This paper studies notions of faithfulness and discusses differences between 
the 
two frameworks, e.g., the fact that starting from small conditioning sets may 
not always be preferable in the stochastic process-setting.

The use of score-based methods, or methods that aggregate the information 
across edges, is an interesting topic for future research.

\subsubsection*{Acknowledgements}

This work was supported by a DFF-International Postdoctoral
Grant (0164-00023B) from Independent Research Fund Denmark. The author is a 
member of the
ELLIIT Strategic Research Area at Lund University.

\bibliographystyle{plainnat}
\bibliography{C:/Users/swmo/Desktop/-/forsk/references}

\appendix

\clearpage
\onecolumn

\section*{Supplementary Material for \\ Faithful graphical representations of 
local independence}

\section{Linear Hawkes Processes}
\label{app:hawkes}

\begin{exmp}[Linear Hawkes process]
	\emph{Linear Hawkes processes} are a class of point processes. A 
	(multivariate) point process, $X_t = (X_t^1,\ldots,X_t^n)$, 
	consists of a set of events, $(t,\alpha)$, such that $t$ is a time 
	point and $\alpha \in V = \{1,2,\ldots,n\}$ is a coordinate process. Point 
	processes may be described using the conditional intensity which we will 
	denote $\lambda_t = (\lambda_t^1,\ldots,\lambda_t^n)$. It holds that
	
	\begin{align*}
	\lambda_t^\beta = \frac{1}{h} \lim_{h\downarrow 0} P(\text{there is a 
	$\beta$-event 
		in } (t,t+h] \mid \mathcal{F}_t^V),
		\end{align*}
		
		and $\lambda_t^\beta$ can therefore be interpreted as describing how 
		likely 
		it is to observe a $\beta$-event in the immediate future given the past 
		of 
		the process. A point process is a \emph{linear Hawkes process} if for 
		all 
		$\beta$
		
		\begin{align*}
		\lambda_t^\beta = \mu_\beta + \sum_{\alpha\in V}\sum_{(s,\alpha):s< t} 
		f_{\beta\alpha}(t - s)
		\end{align*}
		
		where $\mu_\beta$ is a nonnegative constant, $f_{\beta\alpha}$ is a 
		nonnegative function, and the sum is over all events of type $\alpha$ 
		until 
		time $t$. In this example, the $\lambda$-process in Definition 
		\ref{def:li} 
		can be 
		chosen as the conditional intensity.
		
		When the function $f_{\beta\alpha}$ is zero there is no direct 
		dependence 
		of $\lambda_t^\beta$ on the past of the $\alpha$-process. We construct 
		a 
		graph with nodes $V = \{1,2,\ldots,n \}$ such that for $\alpha,\beta\in 
		V$, 
		we include $\alpha\rightarrow\beta$ if and only $f_{\beta\alpha} \neq 
		0$. This 
		graph encodes a set of local independences as described by the global 
		Markov 
		property (Definition \ref{def:globalMarkov}). When the graph is 
		unknown, we 
		can use tests of local independence to learn about the graph. We will 
		say that 
		the graph defined above
		is the \emph{causal graph}, see also Subsection \ref{sssec:cauInter}.
		
		\label{exmp:hawkes}
		\end{exmp}

\section{Asymmetric Graphoids}
\label{sec:asym}

\emph{Graphoid properties} are often used in the context of graphical models of 
random variables \citep{lauritzen1996}. Analogously, \emph{asymmetric graphoid 
properties} may be 
defined \citep{didelezUAI2006,mogensenUAI2018}. These have \emph{left} and 
\emph{right} versions as \emph{symmetry}, 
$(A,B,C) \in \mathcal{I} \Rightarrow (B,A,C)\in \mathcal{I}$, is not assumed.

\begin{defn}[Asymmetric graphoid properties]
	Let $\mathcal{I}$ be an independence model over $V$. We say that 
	$\mathcal{I}$ satisfies \emph{left decomposition} if 
	
	\begin{align*}
		(A,B,C) \in \mathcal{I} \Rightarrow (D,B,C)\in \mathcal{I} \text{ 
		whenever } D\subseteq A.
	\end{align*}
	
	We say that $\mathcal{I}$ satisfies \emph{right decomposition} if 
	
	\begin{align*}
	(A,B,C) \in \mathcal{I} \Rightarrow (A,D,C)\in \mathcal{I} \text{ 
		whenever } D\subseteq B.
	\end{align*}
	
	We say that $\mathcal{I}$ satisfies \emph{left weak union} if 
	
	\begin{align*}
		(A,B,C) \in \mathcal{I} \Rightarrow (A,B,C \cup D)\in \mathcal{I} 
		\text{ 
			whenever } D\subseteq A.
	\end{align*}
	
	We say that $\mathcal{I}$ satisfies \emph{left contraction} if 
	
	\begin{align*}
	(A,B,C) \in \mathcal{I},\  (D,B,A\cup C) \in \mathcal{I}  \Rightarrow (A 
	\cup 
	D,B,C)\in \mathcal{I}.
	\end{align*}
	
	\label{def:asym}
\end{defn}

The following is similar to results in \cite{didelezUAI2006} and 
\cite{mogensenUAI2018}.

\begin{prop}
	Let $\mathcal{I}$ be a local independence model, or a Granger causality 
	model, i.e., an independence model constructed using Definition 
	\ref{def:li} or Definition \ref{def:granger}. The independence model 
	$\mathcal{I}$ satisfies left and right decomposition.
\end{prop}

\begin{proof}
	Left and right decomposition follow immediately from the definitions.
\end{proof}

\section{Detection of Faithfulness Violations}
\label{app:detectable}

Assume that $\mathcal{I}$ is not faithful with respect to the causal graph 
$\mathcal{D}$. We say that a failure of faithfulness is \emph{detectable} if 
there is no other graph, $\mathcal{D}'$, such that $\mathcal{I} = 
\mathcal{I}(\mathcal{D}')$, i.e., no other graph, $\mathcal{D}'$, such that 
$\mathcal{I}$ is Markov and faithful with respect to $\mathcal{D}'$ 
\citep{zhang2008detection}. Detectability implies that we, in principle and for 
infinite 
data, will realize that the independence model we are observing is not 
graphical.

\begin{thm}
	If we assume Markovness and causal minimality, and the faithfulness 
	assumption fails, then the failure is detectable.
	\label{thm:detectable}
\end{thm}

\begin{proof}
	Assume that $\mathcal{I}(\mathcal{D}) \subsetneq \mathcal{I}$. If the 
	failure 
	is undetectable, there exists $\mathcal{D}'$ such that 
	$\mathcal{I}(\mathcal{D}') = \mathcal{I}$. In this case, 
	$\mathcal{I}(\mathcal{D}) \subsetneq \mathcal{I}(\mathcal{D}')$. 
	Proposition \ref{prop:subIndepSubGraph} gives $\mathcal{D}' \subseteq 
	\mathcal{D}$, and therefore $\mathcal{D}' \subsetneq 
	\mathcal{D}$. However, this is a violation of causal minimality. 
	Alternatively, this follows also from uniqueness in Proposition 
	\ref{prop:cauMinUnique}.
\end{proof}

\begin{prop}
	Let $\mathcal{D}_1 = (V,E_1)$ and $\mathcal{D}_2 = (V,E_2)$ be directed 
	graphs. If $\mathcal{I}(\mathcal{D}_1) \subseteq 
	\mathcal{I}(\mathcal{D}_2)$, then $\mathcal{D}_2 \subseteq \mathcal{D}_1$. 
	\label{prop:subIndepSubGraph}
\end{prop}

\begin{proof}
	Let $\alpha\rightarrow\beta$ be an edge which is not in $\mathcal{D}_1$. In 
	this case, $(\alpha,\beta,\pa_{\mathcal{D}_1}(\beta))\in 
	\mathcal{I}(\mathcal{D}_1)$ and 
	$\pa_{\mathcal{D}_1}(\beta)\subseteq V\setminus \{\alpha \}$. Therefore, 
	$(\alpha,\beta,\pa_{\mathcal{D}_1}(\beta))\in \mathcal{I}(\mathcal{D}_2)$, 
	and $\alpha\rightarrow\beta$ is not in $\mathcal{D}_2$ as $\alpha\notin 
	\pa_{\mathcal{D}_1}(\beta)$.
\end{proof}

\section{Edge Order}
\label{sec:edgeOrder}

We define the order of the pair $(\alpha,\beta)$.

\begin{defn}[Order of an ordered pair of nodes]
	Let $\mathcal{D} = (V,E)$ be a directed graph, and $\alpha,\beta\in V$, 
	$\alpha\neq \beta$. The order of $(\alpha,\beta)$ relative to $\mathcal{D}$ 
	is $\inf\{\vert C\vert : (\alpha,\beta,C)\in \mathcal{I}(\mathcal{D}), 
	C\subseteq V\setminus \{\alpha\} \}$, 
	and we denote this by $o(\alpha,\beta,\mathcal{D})$.
\end{defn}

By convention $o(\alpha,\beta,\mathcal{D}) = \infty$ if and only if there is no 
set $C\subseteq V\setminus \{\alpha\} $ such that $ (\alpha,\beta,C)\in 
\mathcal{I}(\mathcal{D})$. The \emph{order} of a graph, 
$\mathcal{D} = (V,E)$, is the largest, finite order 
$o(\alpha,\beta,\mathcal{D})$ if such a finite order exists.

\begin{prop}
	If $o(\alpha,\beta,\mathcal{G}) < \infty$,  then 
	$o(\alpha,\beta,\mathcal{G})\leq \vert \pa_\mathcal{D}(\beta)\vert $
\end{prop}

\begin{proof}
	If $o(\alpha,\beta,\mathcal{G}) < \infty$, then $\alpha\rightarrow\beta$ is 
	not in $\mathcal{D}$, and therefore $\beta$ is $\mu$-separated from 
	$\alpha$ given $\pa_\mathcal{D}(\beta)\subseteq V\setminus \{\alpha\} $.
\end{proof}

\section{The dSGS-algorithm}
\label{app:dSGS}

\begin{algorithm}
	\SetKwInOut{Input}{input}
	\SetKwInOut{Output}{output}
	\Input{$\mathcal{I}$ over $V$ such that $\vert V \vert = n$, and $k$ such 
		that $0\leq 
		k \leq n-1$}
		$i\gets 0$\;
		\For{$i = 0,1,\ldots, k$}{
			\For{$\beta \in V$}{
				\For{$\alpha \in V\setminus \{\beta\}$}{
					\For{$C \subseteq V\setminus \{\alpha\}: \vert C\vert = k$}{
						\If{$(\alpha,\beta,C)\in \mathcal{I}$}
						{
							$E \gets E \setminus \{\alpha\rightarrow \beta \}$\;
							$\mathcal{D} \gets (V,E)$\;
							}	}}
							}
							}
							$\mathcal{D}_{sgs} \gets \mathcal{D}$\;
							\Output{$\mathcal{D}_{sgs}$}
							\caption{Dynamical SGS (dSGS)}
							\label{algo:dSGS}
							\end{algorithm}
							
							The following 
							result uses only the Markov assumption. This is 
							analogous to the 
							SGS-algorithm for DAGs \citep{spirtes1993}. The 
							order 
							of a graph is defined in Appendix 
							\ref{sec:edgeOrder}.
							
							\begin{prop}
								Assume that $\mathcal{D}$ is of order less than 
								or equal to 
								$m$, and that $\mathcal{I}$ 
								is Markov with respect to $\mathcal{D}$. In the 
								oracle case, the output of 
								Algorithm \ref{algo:dSGS} (dSGS), using $k = m$ 
								as the integer parameter, 
								is  a 
								subgraph of $\mathcal{D}$.
								\label{prop:dSGSsubgraph}
								\end{prop}
								
								\begin{proof}
									Assume $\alpha\rightarrow\beta$ is not 
									$\mathcal{D}$. In this case, 
									$(\alpha,\beta,\pa_\mathcal{D}(\beta))\in 
									\mathcal{I}(\mathcal{D})$, $\alpha\notin 
									\pa_\mathcal{D}(\beta)$. We have 
									$o(\alpha,\beta,\mathcal{D}) < \infty$, 
									and therefore $o(\alpha,\beta,\mathcal{D}) 
									\leq m$ by assumption. There exists a $C$, 
									$\vert C\vert \leq m$, and $C\subseteq 
									V\setminus \{\alpha \}$ such that 
									$(\alpha,\beta,C)\in 
									\mathcal{I}(\mathcal{D})$. Using the Markov 
									property, $(\alpha,\beta,C)\in 
									\mathcal{I}$. In the 
									oracle case, this edge is 
									therefore removed using the set $C$, and 
									$\alpha\rightarrow\beta$ is not in
									$\mathcal{D}_{sgs}$.
									\end{proof}

\section{Learning and Faithfulness}
\label{app:learnFaith}

In this section, we relate the contents of Section \ref{sec:characFaith} to 
structure learning.

\subsection{Learning the Edge-Transitive Graph}
\label{ssec:learnEdgeTransitive}

From a collection of test results, i.e., an empirical independence model, we 
may output the corresponding edge-transitive 
graph. This graph is defined for \emph{any} independence model and Algorithm 
\ref{algo:edgeTransitive} (for $k = n$) therefore outputs a graph which is 
faithful 
to the observed test results, regardless of whether there are statistical 
errors in 
the test results.

We say that $\mathcal{I}$ is $k$-faithful with respect to $\mathcal{D}$ if for 
all $C$ such that $\vert C\vert \leq k$

\begin{align*}
(A,B,C) \in \mathcal{I} \Rightarrow (A,B,C) \in \mathcal{I}(\mathcal{D}).
\end{align*}

One should note that $n$-faithfulness, and $(n-1)$-faithfulness, is the same as 
faithfulness. The 
idea of $k$-faithfulness is similar in nature to 
\cite{mogensen2023weak} which defines a weak notion of Markov equivalence by 
restricting the size of the conditioning sets.

\begin{algorithm}
	\SetKwInOut{Input}{input}
	\SetKwInOut{Output}{output}
	\Input{$\mathcal{I}$ over $V$ such that $\vert V \vert = n$, and $k$ such 
		that $0\leq 
		k \leq n-1$}
	$i\gets 0$\;
	$\mathcal{D}$ is the complete graph on nodes $V$\;
	\For{$i = 0,1,\ldots, k$}{
		\For{$\beta \in V$}{
			\For{$\alpha \in V\setminus \{\beta\}$}{
				\For{$C \subseteq V\setminus \{\alpha\}: \vert C\vert = k$}{
					\If{E0, E1, E2, or E3 is violated}
					{
						$E \gets E \setminus \{\alpha\rightarrow \beta \}$\;
						$\mathcal{D} \gets (V,E)$\;
					}	}}
				}
			}
			\Output{$\mathcal{D}$}
			\caption{Edge-Transitive Graph}
			\label{algo:edgeTransitive}
		\end{algorithm}
		
		The next proposition follows immediately from the proof of Theorem 
		\ref{thm:edgeFaith}.
		
		\begin{prop}
			The output of Algorithm \ref{algo:edgeTransitive} is $k$-faithful.
		\end{prop}

		\subsection{Trimming the Output of a Learning Algorithm}
		\label{ssec:trim}
		
		Constraint-based learning algorithms proceed by testing a number of 
		conditional 
		independences. These tests results may be reused to check conditions 
		D0, 
		D1, D2, 
		and D3, and remove any edges that are in violation.
		
		\begin{algorithm}
			\SetKwInOut{Input}{input}
			\SetKwInOut{Output}{output}
			\Input{$\mathcal{I}$ over $V$ such that $\vert V \vert = n$, and a 
			graph 
				$\mathcal{D}= (V,E)$}
			\For{$\beta \in V$}{
				\For{$\alpha \in V\setminus \{\beta\}$}{
					\For{$C \subseteq V\setminus \{\alpha\}$}{
						\If{D0, D1, D2, or D3 is violated}
						{
							$E \gets E \setminus \{\alpha\rightarrow \beta \}$\;
							$\mathcal{D} \gets (V,E)$\;
						}	}}
					}
					$\mathcal{D}_{tr}\gets \mathcal{D}$\;
					\Output{$\mathcal{D}_{tr}$}
					\caption{Trimming}
					\label{algo:trimFaith}
				\end{algorithm}

				\begin{prop}
					The independence model $\mathcal{I}$ (input in Algorithm 
					\ref{algo:trimFaith}) is faithful with respect to 
					$\mathcal{D}$ (output in 
					Algorithm \ref{algo:trimFaith}).
				\end{prop}
				
				A constraint-based learning algorithm need not test all 
				possible independences. 
				If we let $\mathcal{I}$ be an `empirical' independence model, 
				i.e., a set of 
				local independences that are believed to hold/not hold based on 
				statistical tests, we may not 
				have access to 
				the entire $\mathcal{I}$ after running a learning algorithm. In 
				that case, the 
				trimming would need to be restricted to the observed part of 
				$\mathcal{I}$.
				
				\begin{proof}
					Let $\mathcal{D}_{tr}$ be the output of Algorithm 
					\ref{algo:trimFaith}. If 
					$\alpha\rightarrow\beta$ is in $\mathcal{D}_{tr}$, then 
					this edges 
					satisfies the 
					conditions in Definition \ref{def:transitivity} for some 
					graph 
					$\mathcal{D}_i$ such that $\mathcal{D}_{tr}\subseteq 
					\mathcal{D}_i 
					\subseteq \mathcal{D}$. Therefore, the conditions are also 
					satisfied for 
					$\mathcal{D}_{tr}$, and by Theorem \ref{thm:characFaith}, 
					$\mathcal{I}$ is 
					faithful with respect to $\mathcal{D}_{tr}$.
				\end{proof}

\section{Partial Observation}
\label{app:partialObs}

In this section, we will turn our attention to the setting where we only assume 
\emph{partial observation}. We will assume that there exists an underlying 
causal graph, $\mathcal{D}=  (V,E)$, however, we only observe the coordinate 
processes in the set $O$, $O\subseteq V$. In this case, one can use a so-called 
\emph{latent projection} to compute a \emph{directed mixed graph}, 
$\mathcal{G}$, such that $\mathcal{I}(\mathcal{G}) = 
\mathcal{I}(\mathcal{D})_O$ where $\mathcal{I}(\mathcal{D})_O = \{(A,B,C)\in 
\mathcal{I}(\mathcal{D}): A,B,C\subseteq O \}$ \citep{Mogensen2020a}. A 
directed mixed graph may have both directed, $\rightarrow$, and bidirected 
edges, $\leftrightarrow$. In case of partial observation, we only have access 
to $\mathcal{I}_O = \{(A,B,C)\in \mathcal{I}: 
A,B,C\subseteq O \}$ as we can only test local 
independence among the observed coordinate processes. It is important to note 
that the \emph{partial observation} in this paper 
refers to the fact that some coordinate processes are fully unobserved.

A first observation is the fact that the corresponding induced local 
independence 
graph, $\mathcal{D}_{I_O}$, may still be useful, even if its interpretation is 
slightly different. Assuming the equivalence of pairwise and global Markov 
properties, we still have $\mathcal{I}(\mathcal{D}_{I_O}) \subseteq 
\mathcal{I}_{O}$ such that $\mu$-separation in the induced local indepedence 
graph implies local independence.

\begin{prop}
	Let $\mathcal{D} = (V,E)$, and let $\mathcal{G} = (O,E)$ be the latent 
	projection of $\mathcal{G}$ over $O$, $O \subseteq V$. If $\mathcal{I}$ is 
	faithful (ancestor faithful) with respect 
	to $\mathcal{D}$, then $\mathcal{I}_O$ is faithful (ancestor faithful) with 
	respect 
	to $\mathcal{G}$.
\end{prop}

\begin{proof}
	If $\mathcal{I}$ is faithful with respect to $\mathcal{D}$, then 
	$\mathcal{I}_O$ is clearly faithful with respect to $\mathcal{G}$ 
	using the fact that $\mathcal{I}(\mathcal{G}) = \mathcal{I}(\mathcal{D})_O$.
	
	Assume $\mathcal{I}$ is ancestor faithful with respect to $\mathcal{D}$. If 
	there is a directed path from $A$ to $B$ in $\mathcal{G}$ which is 
	$\mu$-connecting given $C$, then there is also a directed path from 
	$A$ to $B$ in $\mathcal{D}$ which is $\mu$-connecting given 
	$C$, and we see that $(A,B,C)\notin \mathcal{I}_O$.
\end{proof}

On the other hand,´parent faithfulness or causal minimality of $\mathcal{I}$ 
and $\mathcal{D}$ is not inherited by $\mathcal{I}_O$ and $\mathcal{G}$ in this 
way.

Note that the next proposition does not assume causal sufficiency.

\begin{prop}[\cite{mogensen2020causal}]
	Assume ancestor faithfulness of 
	$\mathcal{I}$ 
	with respect to $\mathcal{D}$. If $\alpha\rightarrow\beta$ is not in 
	the output of the CS-algorithm (in the oracle case), 
	then $\alpha\rightarrow\beta$ is not in the latent projection of the causal 
	graph, 
	$\mathcal{G}$.
	\label{prop:latentCS}
\end{prop}

The edge $\alpha\rightarrow\beta$, $\alpha\neq\beta$, is in the latent 
projection of $\mathcal{D}$ if and only if there is a directed path from 
$\alpha$ to $\beta$ in $\mathcal{D}$ such that all nonendpoint nodes are 
unobserved, i.e., not in $O$.

\section{Simulations}
\label{app:sim}

We generated data from a VAR(1)-process. We first generated a graph by 
sampling edges independently with a, randomly sampled, success parameter 
between 0 and 
.5. Given the graphical structure, we sampled the nonzero regression 
parameters independently and uniformly on $[-1,1]$. We kept sampling until the 
result was a stable VAR(1)-process. We sampled data from this VAR(1)-process 
(100 
observed time points). We repeated the entire procedure $M$ times (see Figures 
\ref{fig:num} and \ref{fig:numPartial}).

The simulations were implemented in {\tt R} and we used the Granger causality 
test´in the
FIAR package ({\tt condGranger}). Code is available along with this paper.

\subsection{Partial Observation}
\label{ssec:numPart}

We also compare the algorithms in the case of partial observation (Appendix 
\ref{app:partialObs}), also 
including the dFCI-algorithm from \cite{mogensenUAI2018}. This algorithm is the 
only algorithm of the five in Figure \ref{fig:numPartial} which is sound and 
complete in the case of partial observation, i.e., outputs the true graph in 
the oracle case. In the case of partial observation, the learning target is the 
greatest element of the Markov equivalence class of the true graph 
\citep{Mogensen2020a} as the true graph itself is not necessarily identifiable 
from tests of local independence. 
We compare the output of the learning algorithms only to the \emph{directed} 
part of the learning target, i.e., we ignore bidirected edges in the learning 
target.

For this experiment, we sampled the number of unobserved nodes uniformly on 
$\{0,1,\ldots,n\}$. The 
true (and fully observed) graph was then sampled as in Figure \ref{fig:num} 
(see above description). We marginalized the graph using the latent projection 
and computed the 
greatest element of the Markov equivalence class of the latent 
projection.

As seen from Figure \ref{fig:numPartial}, the dFCI does not fare better than 
the simpler algorithms. Most likely this is due to the fact that it uses a 
large number of tests and makes decisions sequentially based on these test 
results. This may lead to propagation of statistical errors.

\begin{figure*}
	\includegraphics[scale = .7]{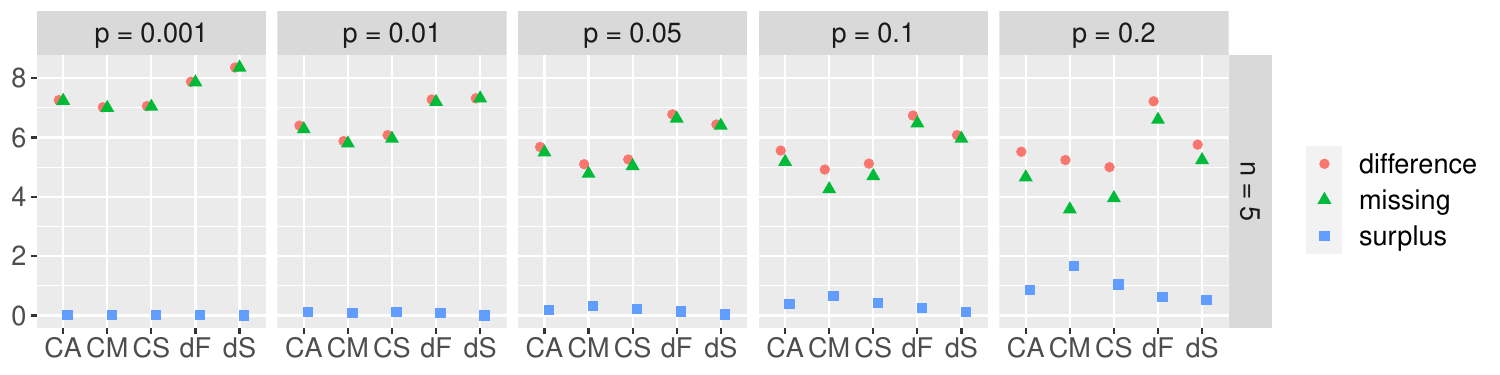}
	\caption{Comparison of algorithms in the case of partial observation. 
	Points 
	indicate mean over $M = 50$ 
		repetitions (see caption of Figure \ref{fig:num} for a description of 
		the symbols). Subsection 
		\ref{ssec:numPart} provides more details.}
		\label{fig:numPartial}
		\end{figure*}

\subsection{Number of Tests}
\label{ssec:numTests}

Of the algorithms reported, only the CA-, the CS-, and the dFCI-algorithms are 
`adaptive' 
in the sense 
that they use different numbers of tests depending on the test results. The 
dSGS-algorithm uses all possible tests, the CM-algorithm uses a single test for 
each 
ordered pair of nodes, and the CS-algorithm uses at most two tests for each 
ordered 
pair of nodes. In the experiment reported in Figure \ref{fig:num}, the number 
of tests used by the CA-algorithm was in the ranges $20-142$ ($n = 5$),  
$45-621$ ($n = 7$), and  $79-1509$ ($n = 9$), respectively.

\section{Proofs}
\label{app:proofs}

\begin{proof}[Proof of Proposition \ref{prop:graphTransitive}]
	We should show that the conditions D0-D3 in Definition 
	\ref{def:transitivity} hold for every $C\subseteq V$ when $\mathcal{I} = 
	\mathcal{I}(\mathcal{D})$. D0 holds as $\alpha\rightarrow\beta$ is 
	$\mu$-connecting for all $C$ such that $\alpha\notin C$. In D1, if 
	$(\gamma,\alpha,C)\notin \mathcal{I}(\mathcal{D})$, then there is a 
	$\mu$-connecting walk from $\gamma$ to $\alpha$ given $C$, and if 
	$\alpha\notin C$, then the composition of this walk with the edge 
	$\alpha\rightarrow\beta$ is $\mu$-connecting from $\gamma$ to $\beta$ 
	given $C$. Conditions D2 and D3 follow similarly. This is clear from the 
	definition of $\mu$-separation.
\end{proof}

\begin{lem}[\cite{Mogensen2020a}]
	If there is a $\mu$-connecting walk from $\alpha$ to $\beta$ given $C$, 
	then there is a $\mu$-connecting walk from $\alpha$ to $\beta$ given $C$ 
	such that all colliders are in $C$.
	\label{prop:allCollInC}
\end{lem}

\begin{proof}[Proof of Proposition \ref{prop:transitiveInclusion}]
	Assume $\alpha\rightarrow\beta$ is in $\mathcal{D}$. If $\alpha\notin C$, 
	then 
	$(\alpha,\beta,C)\notin \mathcal{I}_2$, and therefore 
	$(\alpha,\beta,C)\notin 
	\mathcal{I}_1$. The other conditions follow similarly using the fact that 
	$\mathcal{I}_1\subseteq \mathcal{I}_2$.
\end{proof}

We use the notation $\alpha \sim \beta$ to indicate an edge, 
$\alpha\rightarrow\beta$ or $\alpha\leftarrow\beta$, between nodes $\alpha$ and 
$\beta$.

\begin{proof}[Proof of Theorem \ref{thm:characFaith}]
	Assume first that $\mathcal{D}$ is transitively closed with respect to 
	$\mathcal{I}$, and assume $(A,B,C)\notin \mathcal{I}(\mathcal{G})$. Let 
	$\tilde{\omega}$ be a $\mu$-connecting walk from $\gamma\in A$ to 
	$\delta\in B$ given $C$. We can find a walk, $\omega$, which is 
	$\mu$-connecting from 
	$\gamma$ to $\delta$ given $C$ such that all colliders on $\omega$ are in 
	$C$ 
	(Proposition \ref{prop:allCollInC}). If 
	$\omega$ has length 1, then 
	$\gamma\rightarrow\delta$, $\gamma\notin C$, 
	and 
	from D0 we have $(\gamma,\delta,C)\notin \mathcal{I}$. Otherwise, the walk 
	has a nonendpoint node, $\varepsilon$, $\gamma \sim \ldots \sim 
	\varepsilon \rightarrow \delta$, and $\omega$ is of one of the three types 
	in 
	Lemma \ref{lem:walkTypes}. If it is of type 1, then there is a 
	$\mu$-connecting walk from $\gamma$ to $\varepsilon$ given $C$ and 
	$\varepsilon\notin C$ as $\omega$ is $\mu$-connecting. D1 gives that 
	$(\gamma,\delta,C)\notin \mathcal{I}$. If $\omega$ is of type 2, there is 
	an edge $\alpha\rightarrow\beta$ on $\omega$ such that $\beta$ is in $C$ 
	(all colliders on $\omega$ are in $C$), $\alpha\notin C$, the subwalk from 
	$\gamma$ to $\beta$ is $\mu$-connecting given $C$, and the subwalk from 
	$\alpha$ to $\delta$ is $\mu$-connecting given $C$. D2 gives that 
	$(\gamma,\delta,C)\notin \mathcal{I}$. If $\omega$ is of type 3, there must 
	be a head at $\gamma$, $\gamma \leftarrow \alpha \leftarrow \ldots 
	\leftarrow 
	\varepsilon \rightarrow \delta$. The subwalk from $\alpha$ to $\delta$ is 
	$\mu$-connecting given $C$ as $\alpha\notin C$, and D3 gives that 
	$(\gamma,\delta,C)\notin\mathcal{I}$. This means that in each case
	$(\gamma,\delta,C)\notin \mathcal{I}$. From the left and right 
	decomposition properties of 
	local independence, this means that $(A,B,C)\notin \mathcal{I}$. Note that 
	the 
	right decomposition property is immediate from the definition of local 
	independence/Granger noncausality that we use (see also Section 
	\ref{sec:asym}).
	
	Assume now that $\mathcal{I}$ is faithful with respect to $\mathcal{D}$, 
	that 
	is, $\mathcal{I}\subseteq \mathcal{I}(\mathcal{D})$. Proposition 
	\ref{prop:graphTransitive} gives that $\mathcal{I}(\mathcal{D})$ is 
	transitively closed with respect to $\mathcal{D}$, and Proposition 
	\ref{prop:transitiveInclusion} gives that $\mathcal{I}$ is transitively 
	closed 
	with respect to $\mathcal{D}$.
\end{proof}

For convenience, we say that a $\mu$-connecting walk of length strictly 
greater than 1 is of 
type 1 if $\alpha \dots \rightarrow \gamma \rightarrow \beta$. We say that it 
is 
of type 2 if $\alpha \dots \leftarrow \gamma \rightarrow \beta$ and it contains 
a collider, and we say that it is of type 3 if $\alpha \dots \leftarrow \gamma 
\rightarrow \beta$ and it does not contain a collider. The following lemma 
helps 
clarify the contents of Definition \ref{def:transitivity}: D0-D3 are 
essentially sufficient to characterize the $\mu$-connecting walks.

\begin{lem}
	Any $\mu$-connecting walk of length strictly greater than 1 is of type 
	1, 2, or 3.
	\label{lem:walkTypes}
\end{lem}

\begin{proof}
	Let $\omega$ be a $\mu$-connecting walk of length strictly greater than 
	1. In this case, there is a nonendpoint node, $\gamma$, $\alpha \sim \ldots 
	\sim \gamma 
	\rightarrow \beta$. The statement follows immediately from this.
\end{proof}

The next corollary follows from Theorem \ref{thm:characFaith} and the 
definition of 
faithfulness as $\mathcal{I}(\mathcal{D}_2)\subseteq 
\mathcal{I}(\mathcal{D}_1)$ when $\mathcal{D}_1 \subseteq \mathcal{D}_2$.

\begin{cor}
	Let $\mathcal{D}_1$ and $\mathcal{D}_2$ be graphs such that 
	$\mathcal{D}_1\subseteq \mathcal{D}_2$. If $\mathcal{D}_2$ is transitively 
	closed with respect to $\mathcal{I}$, then $\mathcal{D}_1$ is transitively 
	closed with respect to $\mathcal{I}$.
\end{cor}

\begin{proof}[Proof of Theorem \ref{thm:edgeFaith}]
	Assume $(A,B,C)\notin \mathcal{I}(\mathcal{\mathcal{F}_I})$. In this case, 
	there is a 
	$\mu$-connecting walk from $\alpha\in A$ to $\beta\in B$ given $C$. We 
	can then also find a $\mu$-connecting walk in 
	$\mathcal{F}_I$ from $\alpha\in A$ to 
	$\beta\in B$ given $C$ such that all colliders are in $C$ (Proposition 
	\ref{prop:allCollInC}). We show by induction 
	on walk length that the existence of a $\mu$-connecting walk, $\omega$, 
	from $\gamma$ 
	to $\delta$ given $C$ implies that $(\gamma,\delta,C)\notin \mathcal{I}$, 
	assuming that all colliders on $\omega$ are in $C$. 
	If $\omega$ 
	has length 1, then E0 gives the result. Assume now that it holds for 
	all walks of lengths 1,2,\ldots,$m-1$ and with all colliders in $C$ that 
	$\mu$-connectivity implies dependence. We 
	consider a $\mu$-connecting 
	walk 
	of length $m$. Assume this walk is of type 
	1, say, $\gamma \sim \ldots \rightarrow \varepsilon \rightarrow \delta$. 
	The subwalk from $\gamma$ to $\varepsilon$ 
	is $\mu$-connecting given $C$ and has length $m-1$. From the induction 
	assumption, $(\gamma,\varepsilon,C)\notin \mathcal{I}$. As 
	$\varepsilon\rightarrow\beta$ is in $\mathcal{F}_I$ and $\varepsilon\notin 
	C$, we have that $(\gamma,\delta,C)\notin\mathcal{I}$. If it is of type 2, 
	there is an edge $\alpha\rightarrow \beta$, a $\mu$-connecting walk from 
	$\gamma$ to $\beta\in C$ (all colliders on $\omega$ are in $C$), and a 
	$\mu$-connecting walk from $\alpha\notin C$ to $\delta$ given $C$ such 
	that the $\mu$-connecting walks are both of length less than $m-1$. Using 
	the induction hypothesis and E2, we have $(\gamma,\delta,C)\notin 
	\mathcal{I}$. Finally, if $\omega$ is of type 3, then it follows from 
	similar arguments and E3.
\end{proof}

\begin{proof}[Proof of Proposition \ref{prop:graphicalI}]
	We have $\mathcal{I} = 
	\mathcal{I}(\mathcal{D})$ for a graph $\mathcal{D}$. If $e$ is in 
	$\mathcal{D}$, then it follows from Proposition \ref{prop:graphTransitive} 
	that $e$ is in $\mathcal{F}_I$ by comparing Definitions 
	\ref{def:transitivity} and \ref{def:edgetransitive} and using $\mathcal{I} 
	= \mathcal{I}(\mathcal{D})$. If $e$, say $\alpha\rightarrow\beta$, is 
	not in $\mathcal{D}$, then there exists a set $C$, $\alpha\notin C$, such 
	that $\beta$ is $\mu$-separated from $\alpha$ given $C$ in 
	$\mathcal{D}$, and $(\alpha,\beta,C)\in 
	\mathcal{I}(\mathcal{D})=\mathcal{I}$. Therefore, $e$ is not in 
	$\mathcal{F}_I$ using E0.
\end{proof}

\begin{proof}[Proof of Proposition \ref{prop:hierarchyFaith}]
	The first two implications are obvious from the definitions. 
	
	If $\mathcal{I}$ is Markov and parent faithful with respect to 
	$\mathcal{D}$, 
	we can 
	consider a proper subgraph $\mathcal{D}_0$ of $\mathcal{D}$. There is some 
	edge 
	which is in $\mathcal{D}$, but not in $\mathcal{D}_0$, say 
	$\alpha\rightarrow\beta$, $\alpha\neq \beta$. Using parent faithfulness of 
	$\mathcal{I}$ with respect to $\mathcal{D}$, we have 
	$(\alpha,\beta,C)\notin 
	\mathcal{I}$ for all 
	$C$ such that $\alpha\notin C$. We have 
	$(\alpha,\beta,\pa_{\mathcal{D}_0}(\beta))\in \mathcal{I}(\mathcal{D}_0)$ 
	and 
	$\alpha\notin \pa_{\mathcal{D}_0}(\beta)$ which means that $\mathcal{I}$ is 
	not 
	Markov with respect to $\mathcal{D}_0$. Therefore, $\mathcal{I}$ is 
	causally 
	minimal with respect to $\mathcal{D}$, and we conclude that the combination 
	of 
	Markovness and parent faithfulness 
	implies causal minimality. 
\end{proof}

\begin{proof}[Proof of Proposition \ref{prop:inducedCauMin}]
	By definition of the induced local independence graph, $\mathcal{I}$ 
	satisfies the pairwise Markov property with respect to $\mathcal{D}_I$, and 
	therefore 
	$\mathcal{I}$ is Markov with respect to $\mathcal{D}_I$. Let 
	$\mathcal{D}_0$ 
	be a proper subgraph of $\mathcal{D}_I$, say $\alpha\rightarrow\beta$, 
	$\alpha\neq \beta$, is in $\mathcal{D}$, but not in $\mathcal{D}_0$. In 
	this case, $(\alpha,\beta,V\setminus \{\alpha \}) \in 
	\mathcal{I}(\mathcal{D}_0)$ such that 
	$\mathcal{I}(\mathcal{D}_0)\not\subseteq \mathcal{I}$.
\end{proof}

\begin{proof}[Proof of Proposition \ref{prop:cauMinUnique}]
	If $\alpha\rightarrow\beta$ is not in $\mathcal{D}$, then $\beta$ is 
	$\mu$-separated from $\alpha$ given $V\setminus\{\alpha\}$. If 
	$\alpha\rightarrow\beta$ is in $\mathcal{D}_I$, then 
	$(\alpha,\beta,V\setminus \{\alpha\})\notin \mathcal{I}$ and 
	$\alpha\rightarrow\beta$ must be 
	in $\mathcal{D}$ if $\mathcal{I}(\mathcal{D})\subseteq \mathcal{I}$. Any 
	causally minimal graph is therefore a supergraph of $\mathcal{D}_I$. As 
	$\mathcal{D}_I$ is causally minimal it follows that $\mathcal{D}_I$ is the 
	only such graph.
\end{proof}

\begin{proof}[Proof of Proposition \ref{prop:superCpa}]
	Any $\mu$-connecting walk must have a head into $\beta$, and be of length 
	at 
	least 2, $\alpha \sim \ldots \sim \gamma \rightarrow \beta$. We see that 
	$\gamma\in C$, and therefore $\beta$ is $\mu$-separated from $\alpha$ given 
	$C$. Using Markovness, $(\alpha,\beta,C)\in \mathcal{I}$.
	
	Assume now that $\mathcal{I}$ satisfies left weak union, left 
	decomposition, 
	and left contraction. 
	We have $\pa_\mathcal{D}(\beta) \setminus \{\alpha \}\subseteq C$, and 
	therefore 
	$(V\setminus \pa_\mathcal{D}(\beta), \beta, C \cup \{\alpha\})\in 
	\mathcal{I}$ using the 
	global Markov 
	property and the above argument. 
	If $(\alpha,\beta,C)\in \mathcal{I}$, we use left contraction to obtain 
	$(\{\alpha \} \cup V\setminus \pa_\mathcal{D}(\beta), \beta, C)\in 
	\mathcal{I}$. Using 
	left 
	weak union and left 
	decomposition we obtain 
	$(\alpha,\beta,V\setminus \{\alpha\})\in \mathcal{I}$ using that 
	$\pa_\mathcal{D}(\beta)\subseteq C$. Using the equivalence of pairwise and 
	global Markov properties, we see that $\mathcal{I}$ is Markov with respect 
	to the graph obtained by removing the edge $\alpha\rightarrow\beta$, and 
	this is a violation of 
	causal minimality.
\end{proof}

\begin{proof}[Proof of Lemma \ref{lem:reformulate}]
	The conditions D1' and D3' are weaker than conditions D1 and D3, 
	respectively, so one 
	direction is 
	immediate. We assume that D0, D1', D2, and D3' hold for every edge and set 
	$C$. To show that D1 holds, assume there is a $\mu$-connecting walk from 
	$\gamma$ to $\alpha$ given $C$, and we can choose this walk such that all 
	colliders are in $C$. If this walk is directed, then the result 
	follows immediately. Otherwise, if there is no colliders on the walk, it 
	follows from D3'. If there is a collider, there is some edge on the walk 
	which points 
	towards $\gamma$, and let $\phi \rightarrow \psi$ such that $\psi$ is a 
	collider. We see that the result follows from D2. Condition D3 is shown 
	similarly.
\end{proof}

\begin{proof}[Proof of Theorem \ref{thm:conditionsFaith}]
	Assume that $\mathcal{I}$ is parent faithful with respect to $\mathcal{D}$, 
	and let $\alpha\rightarrow\beta$ be an edge in $\mathcal{D}$. In this case, 
	$(\alpha,\beta,C)\notin \mathcal{I}$, $\alpha\notin C$, and D0 holds. On 
	the other hand, assume that D0 holds. In this case $(\alpha,\beta,C)\notin 
	\mathcal{I}$, $\alpha\notin C$, and it follows from left and right 
	decomposition that 
	$(A,B,C)\notin 
	\mathcal{I}$ for all $A$ and $B$ such that $\alpha\in A$ and $\beta\in B$.
	
	Assume that $\mathcal{I}$ is ancestor faithful with respect to 
	$\mathcal{D}$. In this case, they are also parent faithful, and D0 follows. 
	If there is a directed path from $\gamma$ to $\alpha$ which is 
	$\mu$-connecting given $C$, $\alpha\notin C$, and $\alpha\rightarrow
	\beta$, then 
	$(\gamma,\beta,C)\notin \mathcal{I}$ using ancestor faithfulness. On the 
	other hand, assume that D0 and 
	D1' hold, and that there is a $\mu$-connecting walk from $\alpha$ to 
	$\beta$ given $C$, $\alpha\notin C$. If it has length one, then it follows 
	from D0 that $(\alpha,\beta,C)\notin \mathcal{I}$. Otherwise, it has the 
	form $\alpha 
	\rightarrow \ldots \rightarrow \gamma \rightarrow \beta$, and 
	$(\alpha,\beta,C)\notin \mathcal{I}$ follows from D1' since the subwalk 
	from $\alpha$ 
	to $\gamma$ is $\mu$-connecting given $C$ and $\gamma\notin C$. It 
	follows from left and right decomposition of $\mathcal{I}$ that 
	$(A,B,C)\notin \mathcal{I}$ for all $A$ and $B$ such that $\alpha\in A$ and 
	$\beta\in B$.
	
	Assume that $\mathcal{I}$ is trek faithful with respect to $\mathcal{D}$. 
	It is also parent and ancestor faithful, and D0 and D1' follow. If 
	$\alpha\rightarrow\beta$ and there is a $\mu$-connecting trek from 
	$\alpha$ to $\gamma$ given $C$, then there is also a $\mu$-connecting 
	trek from $\beta$ to $\gamma$ given $C$, and $(\alpha,\beta,C)\notin 
	\mathcal{I}$ 
	using trek faithfulness.  Assume now that D0, D1', and D3' hold, and assume 
	that there is a $\mu$-connecting trek from $\alpha$ to $\beta$ given 
	$C$. If the trek has length one, $(\alpha,\beta,C)\notin \mathcal{I}$ 
	follows from 
	D0. 
	If it is a directed walk, then $(\alpha,\beta,C)\notin \mathcal{I}$ follows 
	from D1'. 
	If it is not a directed walk, then it must have heads at both endpoints 
	such that $\alpha \leftarrow\gamma\sim \ldots  \rightarrow \beta$. There 
	is a $\mu$-connecting trek from $\gamma$ to $\beta$ given $C$ and using 
	D3' gives the result. Again, $(A,B,C)\notin \mathcal{I}$ for all $A$ and 
	$B$ such that $\alpha\in A$ and
	$\beta\in B$.
\end{proof}

\begin{proof}[Proof of Proposition \ref{prop:CSoutput}]
	Under parent dependence, the first step outputs a supergraph of the 
	causal graph, $\mathcal{D}$: If $\alpha\rightarrow\beta$ is in 
	$\mathcal{D}$, then $(\alpha,\beta,\beta)\notin \mathcal{I}$, and 
	this edge is not removed. Let $\mathcal{D}_1$ denote the output of the 
	first step. We have $\pa_\mathcal{D}(\beta)\subseteq 
	\pa_{\mathcal{D}_1}(\beta)$ for all $\beta\in V$. Proposition 
	\ref{prop:superCpa} implies that the second step outputs the causal graph.
\end{proof}

\end{document}